\documentclass[a4paper]{article}

\usepackage{fullpage}

\usepackage[all=normal,bibliography=tight]{savetrees}

\usepackage{tikz}
\usepackage{xspace}
\usepackage{comment}

\usepackage{amsthm,amsmath,amssymb}
\usepackage{enumerate}
\usepackage{fancyvrb, url}

\usepackage{todonotes}

\newtheorem{theorem}{Theorem}
\newtheorem{lemma}[theorem]{Lemma}

\theoremstyle{definition}

\DeclareMathOperator{\treedepth}{td}

\begin{document}

\newcommand{\Oh}{\ensuremath{\mathcal{O}}}

\newcommand{\planarname}{\textsc{Vertex Planarization}}
\newcommand{\planarshort}{\textsc{Vertex Planarization}}
\newcommand{\lmsname}{\textsc{$k \times k$ Permutation Clique}}

\newcommand{\defproblem}[4]{
%  \vspace{1mm}
%  \hline
  \vspace{2mm}
\noindent\fbox{
  \begin{minipage}{0.96\textwidth}
  \begin{tabular*}{\textwidth}{@{\extracolsep{\fill}}lr} #1 & {\bf{Parameter:}} #3 \\ \end{tabular*}
  {\bf{Input:}} #2  \\
  {\bf{Question:}} #4
  \end{minipage}
  }
%  \vspace{1mm}
%  \hline
  \vspace{2mm}
}
\newcommand{\defnoparamproblem}[3]{
%  \vspace{1mm}
%  \hline
  \vspace{2mm}
\noindent\fbox{
  \begin{minipage}{0.96\textwidth}
  #1 \\
  {\bf{Input:}} #2  \\
  {\bf{Question:}} #3
  \end{minipage}
  }
%  \vspace{1mm}
%  \hline
  \vspace{2mm}
}

\title{A tight lower bound for Vertex Planarization on graphs of bounded treewidth}

\author{
  Marcin Pilipczuk\thanks{Institute of Informatics, University of Warsaw, Poland, \texttt{malcin@mimuw.edu.pl}}
}

\date{}

\maketitle

\begin{abstract}
In the \textsc{Vertex Planarization} problem one asks to delete the minimum possible number of vertices from an input graph
to obtain a planar graph.
The parameterized complexity of this problem, parameterized by the solution size (the number of deleted vertices)
has recently attracted significant attention.
The state-of-the-art algorithm of Jansen, Lokshtanov, and Saurabh [SODA 2014]
runs in time $2^{\Oh(k \log k)} \cdot n$ on $n$-vertex graph with a solution of size $k$.
It remains open if one can obtain a single-exponential dependency on $k$
in the running time bound.

One of the core technical contributions of the work of Jansen, Lokshtanov, and Saurabh is an algorithm
that solves a weighted variant of \textsc{Vertex Planarization} in time $2^{\Oh(w \log w)} \cdot n$ on graphs
of treewidth $w$. 
In this short note we prove that the running time of this routine is tight under the Exponential Time Hypothesis, even in unweighted graphs and when parameterizing
by treedepth.
Consequently, it is unlikely that a potential single-exponential algorithm for \textsc{Vertex Planarization} parameterized by the solution size
can be obtained by merely improving upon the aforementioned bounded treewidth subroutine.
\end{abstract}

\section{Introduction}
In the \planarname{} problem,
given an undirected graph $G$ and an integer $k$,
our goal is to delete at most $k$ vertices from the graph $G$
to obtain a planar graph.
If $(G,k)$ is a YES-instance to \planarname{}, then we say that $G$
is a \emph{$k$-apex graph}.
Since many algorithms for planar graphs can be easily generalized to 
near-planar graphs --- $k$-apex graphs for small values of $k$ ---
this motivates us to look for efficient algorithms to recognize $k$-apex
graphs. In other words, we would like to solve \planarname{} 
for small values of $k$.

By a classical result of Lewis and Yannakakis~\cite{lewis}, \planarname{}
is NP-hard when $k$ is part of the input.
Since one can check if a given graph is planar in linear time~\cite{planar-check},
\planarname{} can be trivially solved in time $\Oh(n^{k+1})$, where
$n = |V(G)|$, that is, in polynomial time for every fixed value of $k$.
However, such an algorithm is impractical even for small values of $k$;
a question for a faster algorithm brings us to the realms of
\emph{parameterized complexity}.

In the parameterized complexity, every problem comes with a \emph{parameter},
being an additional complexity measure of input instances.
The central notion is a \emph{fixed-parameter algorithm}: an
algorithm that solves an instance
$x$ with parameter $k$ in time $f(k) |x|^{\Oh(1)}$ for some computable
function $f$. 
Such a running time bound, while still super-polynomial (the function $f$ is usually
exponential), is considered significantly better than say
$\Oh(|x|^k)$, as it promises much faster algorithms for moderate values of $k$
and large instances.
We refer to recent textbooks~\cite{ksiazka,DF13} for a more broad introduction
to parameterized complexity.

Due to the aforementioned motivation, it is natural to consider the solution
size $k$ as a parameter for \planarname{}, and ask for a fixed-parameter algorithm.
Since, for a fixed value of $k$, the class of all $k$-apex graphs is closed
under taking minors, the graph minor theory of Roberston and Seymour immediately
yields a fixed-parameter algorithm, but with enormous dependency on the parameter
in the running time bound.\footnote{Formally, this algorithm is non-uniform, that is, it requires an external advice depending on the parameter only.
However, we can obtain a uniform algorithm using the techniques of Fellows and Langston~\cite{FL}.}
The quest for an explicit and faster fixed-parameter algorithm 
for \planarname{} has attracted significant attention
in the parameterized complexity community in the recent years.
First, Marx and Schlotter~\cite{ildi} obtained a relatively simple algorithm,
with doubly-exponential dependency on the parameter and $n^2$ dependency
on the input size in the running time bound.
Later, Kawarabayashi~\cite{kenichi} obtained a fixed-parameter algorithm
with improved linear dependency on the input size, at the cost of 
worse dependency on the parameter.
Finally, Jansen, Lokshtanov, and Saurabh~\cite{jls} developed an algorithm with
running time bound $2^{\Oh(k \log k)} \cdot n$, improving upon
all previous results.

As noted in~\cite{jls}, a simple reduction shows that \planarname{}
cannot be solved in time $2^{o(k)} \cdot n^{\Oh(1)}$ unless the Exponential
Time Hypothesis fails.
Informally speaking, the Exponential Time Hypothesis (ETH)~\cite{eth} asserts
that the satisfiability of $3$-CNF formulae cannot be verified in time
subexponential in the number of variables.
In the recent years, a number of tight bounds for fixed-parameter algorithms
have been obtained using ETH or the closely related Strong ETH;
we refer to~\cite{lms-survey,marx-survey} for an overview.
In this light, it is natural to ask for tight bounds for fixed-parameter algorithms
for \planarname{}. In particular, \cite{jls} asks for a single-exponential (i.e.,
with running time bound $2^{\Oh(k)} n^{\Oh(1)}$) algorithm.

The core subroutine of the algorithm of Jansen, Lokshtanov, and Saurabh,
is an algorithm that solves \planarname{} in time $2^{\Oh(w \log w)} \cdot n$
on graphs of treewidth $w$.
A direct way to obtain a single-exponential algorithm for \planarname{}
parameterized by the solution size would be to improve the running time
of this bounded treewidth subroutine to $2^{\Oh(w)} \cdot n^{\Oh(1)}$.
In this short note we show that such an improvement is unlikely,
as it would violate 
the Exponential Time Hypothesis.

\begin{theorem}\label{thm:main}
Unless the Exponential Time Hypothesis fails,
there does not exist an algorithm that solves \planarname{}
on $n$-vertex graphs of treewidth at most $w$
in time $2^{o(w \log w)} n^{\Oh(1)}$.
\end{theorem}
In fact, our lower bound holds even for a more restrictive parameter
of \emph{treedepth}, instead of treewidth.

While Theorem~\ref{thm:main} does not exclude the possibility
of a $2^{\Oh(k)} n^{\Oh(1)}$-time algorithm for \planarname{},
it shows that to obtain such a running time one needs to circumvent
the usage of bounded-treewidth subroutine on graphs
of treewidth $\Omega(k)$ in the algorithm of Jansen, Lokshtanov, and Saurabh.

The remainder of this paper is devoted to the proof of Theorem~\ref{thm:main}.

\section{Lower bound}\label{sec:lb}

%In this section we show that \planarshort{}, parameterized by the treewidth
%of the input graph, does not admit a $\Ohstar(2^{o(t \log t)})$-time algorithm
%unless the Exponential Time Hypothesis fails.

We base our reduction on the framework for proving superexponential lower bounds
introduced by Lokshtanov, Marx, and Saurabh~\cite{lms:superexp}.
For an integer $k$, by $[k]$ we denote the set $\{1,2,\ldots,k\}$.
Consequently, $[k] \times [k]$ is a $k \times k$ table of elements
with \emph{rows} being subsets of the form $\{i\} \times [k]$, and \emph{columns}
being subsets of the form $[k] \times \{i\}$.
We start from the following auxiliary problem.

\defproblem{\lmsname{}}{An integer $k$ and a graph $G$ with vertex set $[k] \times [k]$.}{$k$}{Is there a $k$-clique in $G$ with exactly one element from each row and exactly one element from
each column?}

As proven in \cite{lms:superexp}, an $2^{o(k \log k)}$-time algorithm for \lmsname{}
would violate ETH. Hence, to prove Theorem~\ref{thm:main}, it suffices to prove the following.
\begin{lemma}\label{lem:red}
There exists a polynomial time algorithm that, given an instance $(G,k)$ of \lmsname{},
outputs an equivalent instance $(H,\ell)$ of \planarshort{}
where the treedepth of the graph $H$ is bounded by $\Oh(k)$.
\end{lemma}
That is, as announced in the introduction,
we in fact prove a stronger variant of Theorem~\ref{thm:main}, refuting
an existence of a $2^{o(w \log w)} n^{\Oh(1)}$-time algorithm for \planarshort{}
parameterized by the treedepth of the input graph.
Recall that the treedepth of a graph $G$, denoted $\treedepth(G)$, is always not smaller than the treewidth of $G$, and satisfies the following recursive 
formula.
\begin{lemma}[\cite{td-bound}]\label{lem:td}
The treedepth of an empty graph is $0$, and the treedepth of a one-vertex graph equals $1$.
The treedepth of a disconnected graph $G$ equals the maximum of the treedepth of the connected components of $G$.
The treedepth of a connected graph $G$ is equal to
$$\treedepth(G) = 1 + \min_{v \in V(G)} \treedepth(G - \{v\}).$$
\end{lemma}
We refer to the textbook~\cite{sparsity} for more information on treedepth.

The rest of this section is devoted to the proof of Lemma~\ref{lem:red}.

\subsection{One-in-many gadget}

We begin with a description of a gadget that allows us to encode a choice among
many options.

Given two vertices $x$ and $y$, by \emph{introducing a $K_5$-edge $xy$}
we mean the following operation: we introduce three new vertices $z_1$, $z_2$ and $z_3$
and make $x, y, z_1, z_2, z_3$ a clique. Note that in every solution to \planarshort{},
at least one of the vertices of the set $\{x,y,z_1,z_2,z_3\}$ needs to be deleted.
As we do not add any more edges incident to any vertex $z_i$, $i=1,2,3$, we may safely
restrict ourselves to solutions to \planarshort{} that contain $x$ or $y$
and do not contain any of the vertices $z_i$, $i=1,2,3$. That is, we treat
the vertices $z_i$ as undeletable vertices, and henceforth by a ``solution to \planarshort{}''
we mean a solution not containing any such vertex.

For an integer $s \geq 1$, we define an \emph{$s$-choice} gadget $C_s$ as follows.
We start with $3s+2$ vertices denoted $a_i$ for $0 \leq i \leq 2s+1$ and
$b_j$ for $1 \leq j \leq s$.
Then, for each $0 \leq i < 2s+1$ we introduce a $K_5$-edge $a_ia_{i+1}$
and for each $1 \leq j < s$ we introduce two $K_5$-edges $b_ja_{2j-1}$ and $b_ja_{2j}$.
Any choice gadget created in the construction
will be attached to the rest of the graph using the vertices $b_j$;
informally speaking, in any optimal solution, exactly one vertex $b_j$ remains undeleted.
We summarize the properties the $s$-choice gadget in the following lemma; see Figure~\ref{fig:red:choice} for an illustration.

\begin{figure}
\begin{center}
\includegraphics{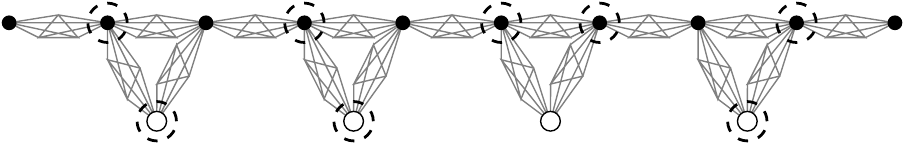}
\caption{Choice gadget $C_4$. The vertices $a_i$ are black and the vertices $b_j$ are white.
  A minimum solution that leaves $b_3$ undeleted is marked with dashed circles.}
  \label{fig:red:choice}
  \end{center}
  \end{figure}

\begin{lemma}\label{lem:red:choice}
For an $s$-choice gadget $C_s$, the following holds.
\begin{enumerate}
\item A minimum solution to \planarshort{} on $C_s$ consists of $2s$ vertices.
\item For every $1 \leq j \leq s$ there exists a minimum solution $X$ to \planarshort{} on $C_s$
that contains all vertices $b_{j'}$ for $j' \neq j$.
\item In every minimum solution to \planarshort{} on $C_s$, at least one vertex $b_j$ remains undeleted.
\item The treedepth of the $s$-choice gadget is $\Oh(\log s)$.
Furthermore, the same treedepth bound holds for a graph constructed from
an $s$-choice gadget by, for every $1 \leq j \leq s$, introducing a
constant-size graph $G_j$ and identifying one vertex of $G_j$
with the vertex $b_j$.
%\item There exists a path decomposition $A_1,A_2,\ldots,A_s$, such that $b_{j'} \in A_j$ if and only if $j=j'$,
%and $|A_j| = 20$ for any $1 \leq j \leq s$.
\end{enumerate}
\end{lemma}
\begin{proof}
First, note that for every $1 \leq j \leq s$, the set
$$\{b_{j'}: j' \neq j\} \cup \{a_{2j'-1}: 1 \leq j' \leq j\} \cup \{a_{2j'}: j \leq j' \leq s\}$$
is a solution to \planarshort{} on $C_s$ of size $2s$ that contains all vertices $b_{j'}$
except for $b_j$.
Moreover, observe that, due to $K_5$-edges, for every $1 \leq j \leq s$, any
solution to \planarshort{} on $C_s$ needs to delete at least two vertices from the set
$\{a_{2j-1}, a_{2j}, b_j\}$, and consequently has size at least $2s$. This settles the first
two claims.

For the third claim, note that any vertex cover of a path of length $2s+1$ needs to
contain at least $s+1$ vertices, and consequently
a solution to \planarshort{} on $C_s \setminus \{b_j: 1 \leq j \leq s\}$ needs to contain at least $s+1$ vertices. Thus, any solution to \planarshort{} on $C_s$ that contains $\{b_j: 1 \leq j \leq s\}$
contains at least $2s+1$ vertices. This settles the third claim.

We prove the last claim by induction on $s$. For $s = \Oh(1)$, the gadget and the attached graphs $G_j$ are of constant size, and the treedepth is constant.
Otherwise, for an $s$-choice gadget $C_s$, we delete the three vertices $z_1,z_2,z_3$ from the $K_5$-edge
between $a_{2\lfloor s/2 \rfloor}$ and $a_{2\lfloor s/2 \rfloor + 1}$. Note that the gadget splits into two connected components,
both being subgraphs of a graph in question constructed from a $\lceil s/2 \rceil$-choice gadget. The treedepth bound $\Oh(\log s)$ follows from Lemma~\ref{lem:td}.
%
%For the last claim, consider a path decomposition $A_j$, $1 \leq j \leq s$, where
% $A_j$ consists of the vertices
%$\{a_{2j-2}, a_{2j-1}, a_{2j}, a_{2j+1}, b_j\}$ and the vertices of all $K_5$-edges
%that have both endpoints in this set. It is straightforward to verify that this is a correct
%path decomposition of $C_s$ and each bag has size $20$.
\end{proof}

\subsection{Construction}
We now give a construction of the \planarshort{} instance $(H,\ell)$, given a \lmsname{}
instance $(G,k)$. Let $m = |E(G)|$ and assume $k \geq 2$.

First, we introduce a \emph{frame} graph $H_F$.
A \emph{ladder of length $n$} is a $2n$-vertex graph that consists of two paths
$v_1,v_2,\ldots,v_n$ and $u_1,u_2,\ldots,u_n$ together with edges $v_iu_i$ for $1 \leq i \leq n$.
A \emph{cycle ladder of length $n$} additionally contains edges $v_nv_1$ and $u_nu_1$, that is,
$v_1,v_2,\ldots,v_n$ and $u_1,u_2,\ldots,u_n$ are in fact cycles.
The frame graph $H_F$ consists of two cycle ladders of length $2k$, with vertex sets
$\{v_i^\Gamma,u_i^\Gamma: 1 \leq i \leq 2k\}$ for $\Gamma \in \{L,R\}$, and
of $k$ ladders of length $4$ with vertex sets $\{v_i^\alpha,u_i^\alpha: 0\leq i \leq 3\}$
for $1 \leq \alpha \leq k$, connected with edges
$v_{2\alpha-1}^Lv_0^\alpha, v_{2\alpha}^Lu_0^\alpha, v_{2\alpha-1}^Rv_{3}^\alpha, v_{2\alpha}^Ru_{3}^\alpha$ for each $1 \leq \alpha \leq k$ (see Figure~\ref{fig:red:frame}).
Note that $H_F$ is $3$-edge-connected and hence has a unique planar embedding.
By $f^\alpha$ we denote the face of the embedding of $H_F$ that is incident to all vertices $v_i^\alpha$, $0 \leq i \leq 3$.

\begin{figure}
\begin{center}
\includegraphics{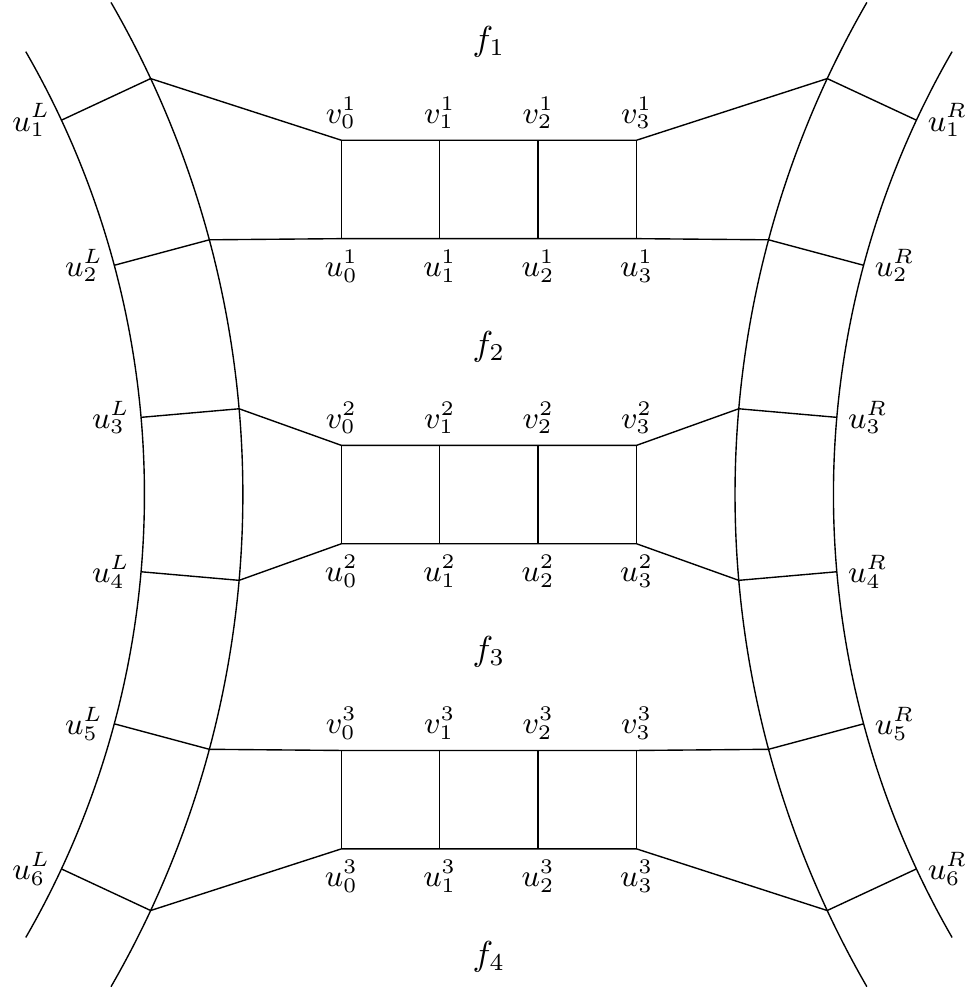}
\caption{Part of the frame graph $H_F$ with its unique embedding.}
  \label{fig:red:frame}
  \end{center}
  \end{figure}

Second, we introduce $k$ vertices $x^\beta$, $1 \leq \beta \leq k$.
Our intention is to ensure that
in any solution to \planarshort{} on $(H,\ell)$, no vertex $x^\beta$ nor
no vertex from the frame $H_F$ will be deleted,
and each vertex $x^\beta$
will be embedded into a different face $f^\alpha$. The choice of which vertex $x^\beta$ is embedded
into which face will correspond to a choice of the vertices of the clique in the instance
$(G,k)$ of \lmsname{}.
We now force such a behavior with some gadgets.

For each $i=1,2,3$, perform the following construction.
For every $1 \leq \alpha,\beta \leq k$, introduce a vertex $y_i^{\alpha,\beta}$ incident
to $v_i^\alpha$ and $x^\beta$.
Moreover, for every $1 \leq \beta \leq k$, introduce a $k$-choice gadget $C^{i,\beta}$
and for each $1 \leq \alpha \leq k$ identify the vertex $b_\alpha$ of $C^{i,\beta}$
with $y_i^{\alpha,\beta}$. See Figure~\ref{fig:red:xv}.
Informally speaking, by the properties of the $k$-choice gadget, each vertex $x^\beta$
needs to select one face $f^\alpha$ that will contain it in the planar embedding.
The fact that the construction is performed three times ensures that no face $f^\alpha$
is chosen by two vertices $x^\beta$, as otherwise a $K_{3,3}$-minor will be left in the graph.

\begin{figure}
\begin{center}
\includegraphics{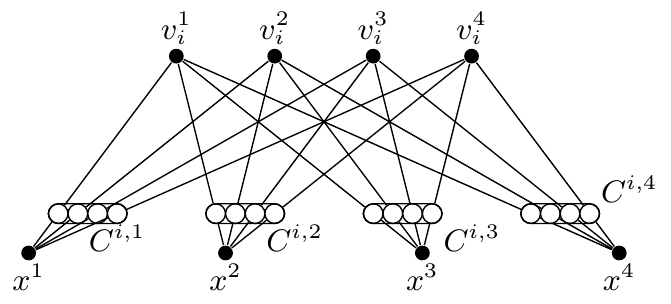}
\caption{Connections between the frame and vertices $x^\beta$ that ensure
that every vertex $x^\beta$ is embedded into a different face $f^\alpha$.
Vertices $y_i^{\alpha,\beta}$ are depicted white.
}
  \label{fig:red:xv}
  \end{center}
  \end{figure}

Let us now move to the description of the encoding of the edges of $G$.
For an edge $e \in E(G)$, let $e = (p(e),\gamma(e))(q(e),\delta(e))$ where $p(e) < q(e)$ (note that
edges $e$ with $p(e)=q(e)$ are irrelevant to the problem, and hence we may assume
there are no such edges).
For $1 \leq p < q \leq k$, we define $E(p,q) = \{e \in E(G) : (p(e),q(e)) = (p,q)\}$.
For each edge $e \in E(G)$, we introduce in $H$ three vertices
$c_e^\vdash, c_e, c_e^\dashv$,
four edges $v_0^{p(e)}c_e^\vdash, x^{\gamma(e)}c_e^\vdash, v_0^{q(e)}c_e^\dashv, x^{\delta(e)}c_e^\dashv$
and two $K_5$-edges $c_e^\vdash c_e$ and $c_e^\dashv c_e$ (see Figure~\ref{fig:red:edge}).
Moreover, for every $1 \leq p < q \leq k$,
introduce a $|E(p,q)|$-choice gadget $\hat{C}^{p,q}$ and for every
$e \in E(p,q)$ identify $c_e^\vdash$ with a distinct vertex $b_j$ of $\hat{C}^{p,q}$.
Informally speaking, for each edge $e$ we need either to delete $c_e^\vdash$ and $c_e^\dashv$
or only $c_e$; however, the second option is only possible if $x^{\gamma(e)}$ is embedded
into $f_{p(e)}$ and at the same time $x^{\delta(e)}$ is embedded into $f_{q(e)}$.
The choice gadget $\hat{C}^{p,q}$ ensures that we can choose the second, cheaper option
only once per each pair $(p,q)$.

\begin{figure}
\begin{center}
\includegraphics{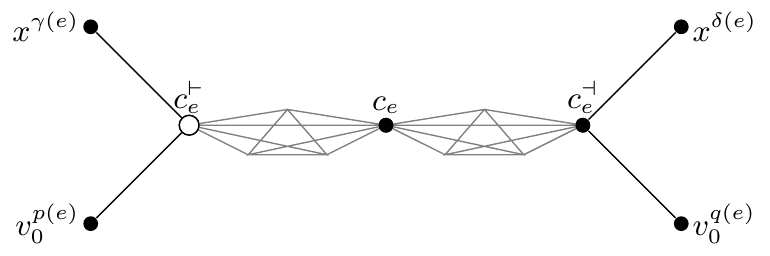}
\caption{Gadget introduced for an edge $e \in E(G)$.
  The vertex $c_e^\vdash$, marked white, is part of the choice gadget $\hat{C}^{p(e),q(e)}$.}
  \label{fig:red:edge}
  \end{center}
  \end{figure}

We set $\ell = 3m + 6k^2$. This completes the description of the instance $(H,\ell)$.
Note that the budget $\ell$ is tight: it allows only to choose a minimum solution
in all introduced choice gadgets, and one endpoint of each $K_5$-edge $c_e^\dashv c_e$.

\subsection{Treedepth bound}

\begin{lemma}
The treedepth of $H$ is $\Oh(k)$.
\end{lemma}
\begin{proof}
We use the recursive formula of Lemma~\ref{lem:td}.
First, we delete from $H$ all vertices of the frame $H_F$
and all vertices $x^\beta$. Note that we have deleted 
only $17k$ vertices in this manner. By Lemma~\ref{lem:td},
it suffices to show that every connected component of the remaining
graph has treedepth $\Oh(k)$; we will in fact show a stronger
bound of $\Oh(\log k)$.

Observe that the remaining graph contains two types of connected components.
The first type are the $k$-choice gadgets $C^{i,\beta}$ for $1 \leq i \leq 3$
and $1 \leq \beta \leq k$; by Lemma~\ref{lem:red:choice}, 
every such gadget has treedepth $\Oh(\log k)$.
The second type are the $|E(p,q)|$-choice gadgets $\hat{C}^{p,q}$
for $1 \leq p < q \leq k$, together with the vertices
$c_e^\vdash, c_e, c_e^\dashv$ and the $K_5$-edges between them.
As $|E(p,q)| \leq k^2$ for every $1 \leq p < q \leq k$, by Lemma~\ref{lem:red:choice} the treedepth
of these connected components is also $\Oh(\log k)$.
This finishes the proof of the lemma.
\end{proof}

\subsection{Equivalence}

In the following two lemmata we show the equivalence of the constructed
\planarshort{} instance $(H,\ell)$ and the input \lmsname{} instance $(G,k)$,
completing the proof of Lemma~\ref{lem:red} and of Theorem~\ref{thm:main}.

\begin{lemma}[Completeness]
If $(G,k)$ is a YES-instance to \lmsname{}, then $(H,\ell)$ is a YES-instance to \planarshort{}.
\end{lemma}
\begin{proof}
Let $\rho: [k] \to [k]$ be a solution to \lmsname{} on $(G,k)$, that is,
$\rho$ is a permutation of $[k]$ and $K := \{(p,\rho(p)): 1 \leq p \leq k\}$
is a clique in $G$. Consider the following set $X \subseteq V(H)$.
\begin{enumerate}
\item For each $i = 1, 2,3$ and $1 \leq \alpha \leq k$, $X$ contains a minimum
solution (i.e., of size $2k$) to \planarshort{} in the gadget $C^{i,\rho(\alpha)}$
that contains all vertices $y_i^{\alpha',\rho(\alpha)}$ for $1 \leq \alpha' \leq k$ except for
$y_i^{\alpha,\rho(\alpha)}$.
\item For each $1 \leq p < q \leq k$, denote by $e(p,q)$ the unique edge
in $E(p,q)$ such that $e(p,q) = (p,\rho(p))(q,\rho(q))$.
Then $X$ contains a minimum solution to \planarshort{} in the gadget $\hat{C}^{p,q}$
that contains all vertices $c_e^\vdash$ for $e \in E(p,q)$
except for $c_{e(p,q)}^\vdash$, the vertex $c_{e(p,q)}$ and all vertices
$c_e^\dashv$ for $e \in E(p,q)$ except for $c_{e(p,q)}^\dashv$.
\end{enumerate}
Note that we have introduced $3 \cdot k \cdot 2k = 6k^2$ vertices in the first step
and $3m$ vertices in the second step. Hence, $|X| = 3m+6k^2 = \ell$.
We now argue that $H \setminus X$ is planar. It suffices to prove it for each
$2$-connected component of $H \setminus X$. Note that the claim is trivial
or follows from Lemma~\ref{lem:red:choice} for each $2$-connected component
of $H \setminus X$ except for the one that contains the frame $H_F$.

Consider the unique planar embedding of $H_F$ and embed into each face $f^\alpha$
the vertex $x^{\rho(\alpha)}$. Note that each vertex $x^\beta$ is embedded into a different face.
It is straightforward to verify that $x^{\rho(\alpha)}$ can be embedded into $f^\alpha$
together with vertices $y_i^{\alpha,\rho(\alpha)}$ for $i=1,2,3$
and the vertices $c_e^\vdash$ or $c_e^\dashv$ that correspond to the edges of $G[K]$
incident to $(\alpha,\rho(\alpha))$.
Note that all other vertices of $H \setminus X$ lie in different $2$-connected components
than $H_F$ and, consequently, $H \setminus X$ is planar.
\end{proof}

\begin{lemma}[Soundness]
If $(H,\ell)$ is a YES-instance to \planarshort{}, then $(G,k)$ is a YES-instance to \lmsname{}.
\end{lemma}
\begin{proof}
Let $X \subseteq V(H)$ be such that $|X| \leq \ell$ and $H \setminus X$ is planar.
By Lemma~\ref{lem:red:choice}, $X$ needs to contain at least $2k$ vertices from each
gadget $C^{i,\beta}$ for $i=1,2,3$ and $1 \leq \beta \leq k$; note that there
are $3k$ such gadgets.
Moreover, $X$ needs to contain at least $2|E(p,q)|$ vertices from each
gadget $\hat{C}^{p,q}$ for $1 \leq p < q \leq k$, which totals to at least $2m$ vertices.
Finally, for each $1 \leq i \leq m$, $X$ needs to contain at least one vertex
of the $K_5$-edge $c_ec_e^\dashv$. As $|X| \leq \ell = 6k^2 + 3m$, we infer that 
$|X| = \ell$ and $X$ contains a minimum solution to \planarshort{} on each introduced
choice gadget, and exactly one vertex from the pair $\{c_e,c_e^\dashv\}$ for each $1 \leq i \leq m$.
In particular, $X$ does not contain any vertex of the frame graph $H_F$, nor any vertex $x^\beta$, $1 \leq \beta \leq k$.

By the properties of the choice gadget $C^{i,\beta}$, there exists $\alpha(i,\beta)$
such that $y_i^{\alpha(i,\beta),\beta} \notin X$. Recall that $X$ does not contain any vertex
of $H_F$, and $H_F$ is $3$-edge-connected and, hence, admits a unique planar embedding
depicted on Figure~\ref{fig:red:frame}. As $X$ does not contain $x^\beta$,
we infer that $\alpha(i,\beta) = \alpha(i',\beta)$ for every $i,i' \in \{1,2,3\}$.
Hence, we may suppress the argument $i$ and henceforth analyze function $\alpha(\beta)$
such that $y_i^{\alpha(\beta),\beta} \notin X$ for every $1 \leq \beta \leq k$ and $i = 1,2,3$.
Note that $x^\beta$ needs to be embedded into the face $f^{\alpha(\beta)}$ of $H_F$ in any
planar embedding of $H \setminus X$.

We now argue that $\alpha(\cdot)$ is a permutation. By contradiction, let $\alpha(\beta) = \alpha(\beta')$ for some $\beta \neq \beta'$. It is straightforward to verify that the following sets form a model of a $K_{3,3}$ minor in $H \setminus X$, contradicting its planarity.
\begin{enumerate}
\item $\{v_i^\alpha, y_i^{\alpha,\beta}, y_i^{\alpha,\beta'}\}$ for $i=1,2,3$;
\item $\{x^{\beta}\}$ and 
$\{x^{\beta'}\}$;
\item $\{u_{1}^\alpha, u_{2}^\alpha, u_{3}^\alpha\}$.
\end{enumerate}
We infer that $\alpha(\cdot)$ is a permutation of $[k]$. We claim 
$K := \{(\alpha(\beta),\beta): 1 \leq \beta \leq k\}$ induces a clique in $G$. That is,
that $\rho := \alpha^{-1}$ is a solution to \lmsname{} on $(G,k)$.

Pick arbitrary $1 \leq p < q \leq k$. Our goal is to prove that $(p,\rho(p))(q,\rho(q)) \in E(G)$.
Consider the choice gadget $\hat{C}^{p,q}$. By Lemma~\ref{lem:red:choice},
there exists $e \in E(p,q)$ such that $c_e^\vdash \notin X$.
Consequently, $c_e \in X$ due to the $K_5$-edge $c_e^\vdash c_e$,
and thus $c_e^\dashv \notin X$.
As $c_e^\vdash$ is adjacent to both $v_0^p$ and $x^{\gamma(e)}$, we infer
that $x^{\gamma(e)}$ is embedded into the face $f^p$ and, consequently,
$\alpha(\gamma(e)) = p$. Symmetrically, we infer that $\alpha(\delta(e)) = q$.
Thus, $e = (p,\rho(p))(q,\rho(q)) \in E(G)$ and the lemma is proven.
\end{proof}

\section*{Acknowledgements}

We thank Tomasz Kociumaka for numerous discussions on the complexity of vertex deletion problems
to minor-closed graph classes. 

Research supported by Polish National Science Centre grant DEC-2012/05/D/ST6/03214.

\bibliographystyle{abbrv}
\bibliography{genus}
\end{document}